\documentclass[lettersize,journal]{IEEEtran}
\IEEEoverridecommandlockouts
\usepackage{graphicx}
\usepackage[table,xcdraw]{xcolor}
\usepackage{multicol}
\usepackage{bmpsize}
\usepackage[cmex10]{amsmath}
\usepackage{algorithm}
\usepackage{algorithmic}
\usepackage{breqn}
\usepackage{listings}
\usepackage{amsfonts}
\usepackage{bm}
\allowdisplaybreaks
\usepackage[utf8]{inputenc}
\usepackage{amssymb}
\usepackage{array}
\usepackage{eqparbox}
\usepackage[tight,footnotesize]{subfigure}
\usepackage{grffile}
\usepackage{multicol}
\usepackage{float}
\usepackage{url}
\usepackage{paralist}
\usepackage{multirow}
\usepackage{wrapfig}
\usepackage{color}
\usepackage{tikz}
\usepackage{mathtools, cuted}
\usepackage{lipsum}
\usetikzlibrary{shapes.geometric, arrows}
\usepackage{cite}
\usepackage{xspace}
\usepackage{atbegshi}
\usepackage{setspace}
\usepackage{subfigure}

\usepackage{amsthm}

\newtheorem{asu}{Assumption}

\newtheorem{lem}{Lemma}
\newtheorem{proposition}{Proposition}

\newcommand{\MATLAB}{\textsc{Matlab}\xspace}

\newcommand{\RN}[1]{%
	\textup{\uppercase\expandafter{\romannumeral#1}}%
}
\def\BibTeX{{\rm B\kern-.05em{\sc i\kern-.025em b}\kern-.08em
    T\kern-.1667em\lower.7ex\hbox{E}\kern-.125emX}}
\begin{document}

\title{On Distributed and Asynchronous Sampling of Gaussian Processes for Sequential Binary Hypothesis Testing}

\author{Nandan Sriranga,~\IEEEmembership{Student Member,~IEEE}, Saikiran Bulusu,~\IEEEmembership{Member,~IEEE}, Baocheng Geng,~\IEEEmembership{Member,~IEEE}, Pramod K. Varshney,~\IEEEmembership{Life Fellow,~IEEE}
\thanks{Nandan Sriranga, Saikiran Bulusu, and Pramod K. Varshney are with the Department of Electrical Engineering and Computer Science, Syracuse University, Syracuse, NY 13244 USA (email:\{nsrirang, sbulusu, varshney\}@syr.edu) }
\thanks{Baocheng Geng is with the Department of Computer Science, University of Alabama at Birmingham, Birmingham, AL 35294 USA (email:\{bgeng\}@uab.edu) }}%

\maketitle

\begin{abstract}
In this work, we consider a binary sequential hypothesis testing problem with distributed and asynchronous measurements. The aim is to analyze the effect of sampling times of jointly \textit{wide-sense stationary} (WSS) Gaussian observation processes at distributed sensors on the expected stopping time of the sequential test at the fusion center (FC). The distributed system is such that the sensors and the FC sample observations periodically, where the sampling times are not necessarily synchronous, i.e., the sampling times at different sensors and the FC may be different from each other. \color{black} The sampling times, however, are restricted to be within a time window and a sample obtained within the window is assumed to be \textit{uncorrelated} with samples outside the window.  We also assume that correlations may exist only between the observations sampled at the FC and those at the sensors in a pairwise manner (sensor pairs not including the FC have independent observations). The effect of \textit{asynchronous} sampling on the SPRT performance is analyzed by obtaining bounds for the expected stopping time. We illustrate the validity of the theoretical results with numerical results.   
\end{abstract}


\section{Introduction}
\IEEEPARstart{S}{ensors} are battery-powered and in many applications operate remotely, and, therefore, energy efficiency is a critical concern in wireless sensor networks (WSNs) \cite{sadler2005fundamentals, appadwedula2005energy}. In such scenarios, minimizing the number of communication exchanges is a crucial design goal. Sequential tests are especially advantageous as they minimize the number of communication exchanges, thereby reducing energy consumption in the system. 
When continuous-time signals are observed at sensors, it is impractical to transmit the signal values at each instant of time. Therefore, the signals are sampled at specific instants of time and are transmitted to the Fusion Center (FC). Efficient sampling at sensors can improve the quality of information with respect to the decision-making problem and also contribute to the reduction in the expected time required for sequential tests.

The notion of sampling signals in distributed systems observing random states that are evolving with time has been considered for state-estimation in target tracking systems by using staggered sensing in \cite{niu2005temporally, liu2012temporally}, and for discrete-time dynamical systems in \cite{liu2014optimal}. Our work is partly motivated by the problems studied in \cite{niu2009sampling} and \cite{yu1997sampling}, both of which consider the effect of sampling a single continuous-time wide-sense stationary (WSS) Gaussian process on the detection performance. The work in \cite{yu1997sampling}, considers non-uniform sequential sampling for a WSS continuous process to improve the detection performance of a binary hypothesis testing problem, and in \cite{niu2009sampling}, the authors propose several sampling schemes for a sequential binary hypothesis testing problem in order to account for the temporal dependence of the observations of a WSS process. Both of these works address the effect of correlation between observations of the same process at different points in time on the detection performance.
 
The sampling of continuous processes for sequential hypothesis testing has also been studied in the distributed sensor network context in  \cite{fellouris2011decentralized} where, a \textit{level-triggered} sampling scheme is proposed to obtain discrete-time samples at sensors that allows for 1-bit data transmissions. Similarly, the authors consider the level-triggered sampling scheme for a distributed sequential spectrum sensing problem in \cite{yilmaz2012cooperative}.  Moreover, distributed sequential hypothesis testing has been studied for sensors observing spatially dependent observations, where the joint probability distribution of sensor observation is modeled using copulas in \cite{zhang2019distributed} and a window of observations is used to perform a truncated sequential test.

In this work, we study the effect of pairwise temporal correlations that may exist between sensor observations and the FC, on sequential hypothesis testing, using the sequential probability ratio test (SPRT), when FC receives sampled observations from the sensors. We determine the effect of the temporal correlations on the expected stopping time of the test and characterize the performance with respect to a fixed sampling time at the FC. Our work focusses on deriving the upper and lower bounds of the stopping times, when the sensor observations are correlated with the FC's observations in a pairwise manner but are uncorrelated with each other, within a finite-duration sampling period. 

\section{Background and problem formulation}
In this section, we motivate the problem and describe the system model, the sampling scheme, and the performance of the sequential test in terms of the expected stopping time of the test.

\subsection{Sampling Scheme}
The sampling scheme is based on the idea of the group sampling scheme in \cite{niu2009sampling}, where groups of samples are formed from a collection of samples at a sensor obtained over time (temporal sampling). In this work, the groups of samples are formed by considering samples across sensors, where each sensor contributes one sample to the group (spatial sampling). 

 We consider the binary sequential hypothesis testing problem in which there are $N$ sensors and an FC, where the $j^{th}$ sensor observes the sample $x_j(t_j+kT)$ where $k \in \mathbb{Z}^{+}$ and $T>0$. For a given $k$, the observations $x_{j}(t_j+kT)$ are observed in a sampling window $kT \leq t_j + kT \leq kT+\epsilon$, for $j \in [N]$. The FC samples the observations $x_{fc}(t_{fc}+kT)$ in the sampling window $kT \leq t_{fc} + kT \leq kT+\epsilon$. This allows us to form a \textit{group} of samples, where the samples corresponding to the $k^{th}$ group are assumed to be \textit{independent} of the samples belonging to all other groups. 
  The SPRT is performed once the sensor samples observed at times $t_j + kT$, $j\in [N]$, arrive at the FC within the interval $[kT, kT + \epsilon]$. 
  The sampling scheme in the parallel network of sensors is illustrated in Fig. \ref{fig:sampling}, where each sensor samples during a sampling window of duration of $\epsilon$, and transmits the sample to the fusion center. 

Due to the assumed WSS property of the observations, the pairwise correlations decrease with time. Therefore, any pair of samples that are widely separated in time, i.e., large T, are assumed to be independent (residual correlations are not considered). \color{black}We consider processes that exhibit \textit{short-lived} temporal correlations termed as \textit{a-dependent} processes \cite{edition2002probability}, where samples that are spaced sufficiently far apart in time from each other, are uncorrelated. 

\begin{figure}
    \centering
    \includegraphics[width=80mm]{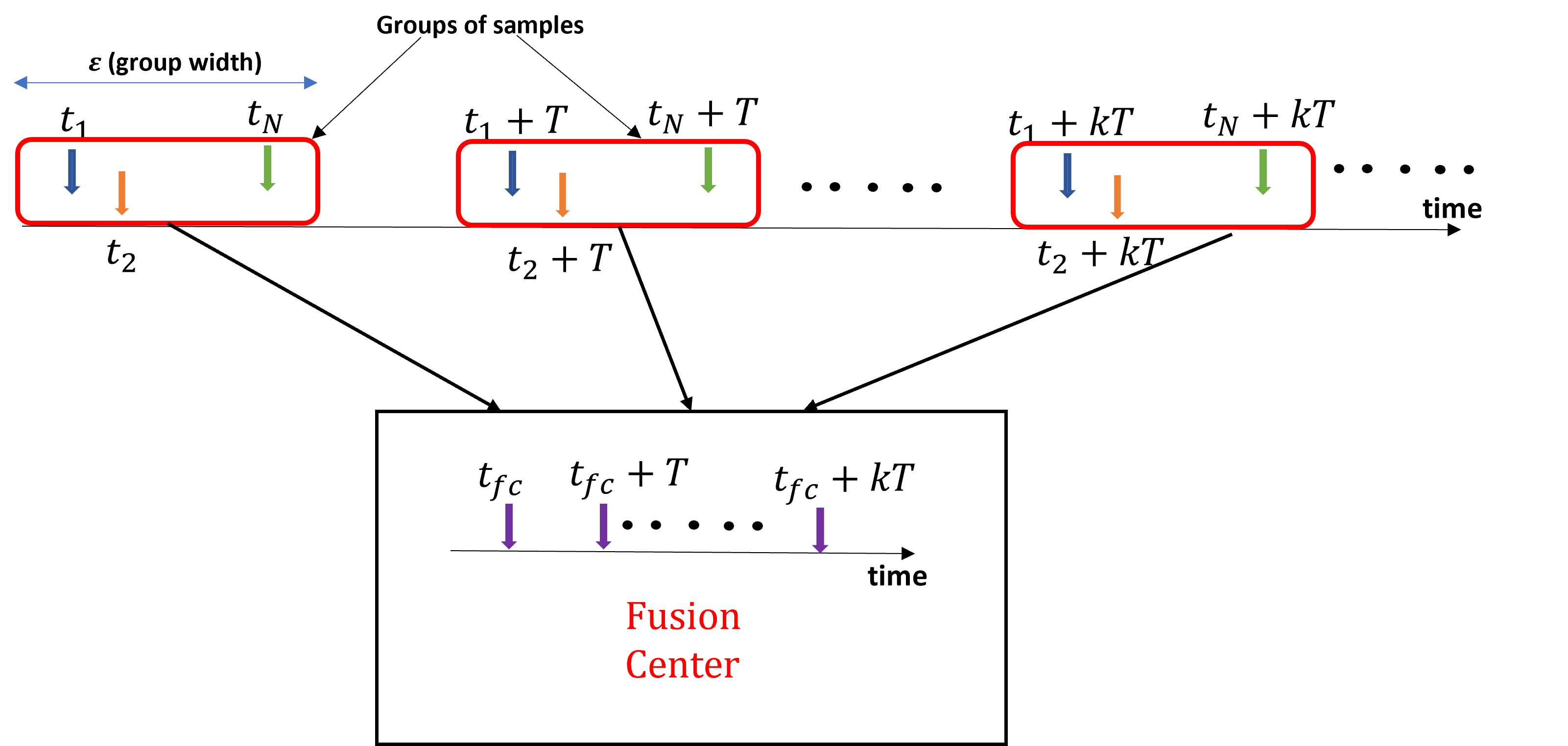}
    \caption{The sampling scheme in this work. The colored arrows indicate the individual samples from different sensors and the red box indicates a group of samples obtained within a sampling window of length $\epsilon$. All sampling times are such that $0< t_1, t_2, \dots, t_{fc} < \epsilon$.} 
    \label{fig:sampling}
\end{figure}

\subsection{Observation Model}

We denote the vector of all the samples acquired during the $k^{th}$ sampling period by $\mathbf{x}(\mbox{\boldmath$t$},k) = [x_1(t_1 + kT) \thinspace \thinspace  x_2(t_2+kT) \thinspace \thinspace \dots \thinspace \thinspace x_N(t_N+kT)]^{\text{T}} \in \mathbb{R}^{N}$, where $\mbox{\boldmath$t$} = [t_1+kT \thinspace \thinspace t_2+kT \thinspace \thinspace \dots \thinspace \thinspace t_N+kT]^{\text{T}}$. The observation model is given as 
\begin{align}
    &H_1: {x}_j(t_j+kT) = {s}_j + \thinspace {n}_j(t_j+kT),  \quad \text{and}\nonumber\\
    &H_0: {x}_j(t_j+kT) = {n}_j(t_j+kT),
    \label{eq:obs_model_sensors}
\end{align}
for $j=1, \dots, N$ and $k \in \mathbb{Z}^{+}$.
We assume that the signals $s_j$ are constant in time and deterministic for $j\in [N]$ and denote the vector of all signals under $H_1$ by $\mathbf{s} = [s_1, s_2, \dots, s_N]^{\text{T}} \in \mathbb{R}^{N}$. In this work, it is assumed that the FC also observes the PoI by sampling the observations at sampling time $t_{fc}+kT$. The observation model at the FC is
\begin{align}
    &H_1: {x}_{fc}(t_{fc}+kT) = {s}_{fc} + \thinspace {n}_{fc}(t_{fc}+kT),  \quad \text{and}\nonumber\\
    &H_0: {x}_{fc}(t_{fc}+kT) = {n}_{fc}(t_{fc}+kT).
    \label{eq:obs_model_FC}
\end{align}
We assume that the noise processes $n_j(t_j+kT)$ and the noise process at the FC $n_{fc}(t_{fc}+kT)$, are jointly WSS. The description of the exact nature of the noise process $\mathbf{n}(\mbox{\boldmath$t$},k)$ considered in this work along with examples to illustrate the noise covariance structure is postponed to a future sub-section.


\subsection{SPRT and Expected Stopping Time}

For the binary hypothesis testing problem with observations $\tilde{x}(\mathbf{t}, 1), \dots, \tilde{x}(\mathbf{t}, k)$ which are independent and identically distributed in time, the SPRT statistic at the $k^{th}$ iteration of the test is the sum of the log-likelihood ratios (LLRs) $\sum_{i=1}^{k} \Lambda(\mathbf{t}, i)$, where $\Lambda(\mathbf{t}, i) = \text{log} \thinspace \frac{f(\tilde{x}(\mathbf{t}, i) | H_1)}{f(\tilde{x}(\mathbf{t}, i) | H_0)}$. The SPRT at stage ${k}$, which corresponds to the $k^{th}$ sampling period is given by
$$
\sum_{i=1}^{k} \Lambda(\mathbf{t}, i) \quad
\begin{cases}
> \Delta_1 , \thinspace \text{decide} \thinspace \thinspace \mathcal{H}_1\\
< \Delta_0 , \thinspace  \text{decide} \thinspace \thinspace \mathcal{H}_0\\
\textit{otherwise}, \thinspace  \text{continue sampling}.
\end{cases}
$$
The log-likelihood ratio for the observation model in II-C is given as 
\begin{equation}
    \Lambda(\mathbf{t}, i) = \tilde{\mathbf{x}}^{\text{T}}(\mathbf{t}, i) \mathbf{\Sigma}^{-1}(\mathbf{t}) \mathbf{\tilde{s}} - \frac{1}{2}\mathbf{\tilde{s}}^{\text{T}} \mathbf{\Sigma}^{-1}(\mathbf{t}) \mathbf{\tilde{s}}
    \label{LLR}
\end{equation}
where $\tilde{\mathbf{x}}^{\text{T}}(\mathbf{t}, i) = [{\mathbf{x}}^{\text{T}}(\mathbf{t}, i) \quad x_{fc}(t_{fc})]$, $\tilde{\mathbf{s}}^{\text{T}} = [{\mathbf{s}}^{\text{T}} \quad s_{fc}]$ and, $\mathbf{\Sigma}^{-1}(\mathbf{t})$ is the inverse covariance matrix of the random variables $x_1(t_1), x_2(t_2), \dots, x_N(t_N) \thinspace \text{and} \thinspace x_{fc}(t_{fc})$. 

The values of the thresholds $\Delta_1$ and $\Delta_0$, depend on the desired levels of probability of false alarm $P_{FA}$ and probability of detection $P_D$ according to Wald's approximations \cite{abraham1947sequential} as 
\begin{align}
    & \Delta_1 \approx \text{log} \thinspace \bigg ( \frac{P_D}{P_{FA}} \bigg), \quad \text{and} \\
    & \Delta_0 \approx \text{log} \thinspace \bigg (\frac{1-P_D}{1-P_{FA}} \bigg).
\end{align}

For pre-specified values of $P_{FA}$ and $P_D$ that are to be achieved, the expected stopping times of the sequential test can be obtained by assuming that the overshoots of the test-statistic when the test halts by crossing the thresholds $\Delta_0$ and $\Delta_1$, are negligible \cite{tartakovsky2014sequential}. Then, the expected stopping times are
\begin{equation}
  \quad  \mathbb{E}[S^{*}(\mathbf{t})|\mathcal{H}_1] \approx \frac{P_{D} \cdot \Delta_1 + (1-P_{D}) \cdot \Delta_0}{\mathbb{E}[\Lambda(\mathbf{t}, i)|\mathcal{H}_1]}, \quad \text{and}
    \label{eq:exp_stop_time_H1}
\end{equation}
\begin{equation}
        \mathbb{E}[S^{*}(\mathbf{t})|\mathcal{H}_0] \approx \frac{P_{FA} \cdot \Delta_1+ (1-P_{FA}) \cdot \Delta_0}{\mathbb{E}[\Lambda(\mathbf{t}, i)|\mathcal{H}_0]}, 
    \label{eq:exp_stop_time_H0}
\end{equation}

where $S^*(\mathbf{t})$ is the stopping time of the sequential test at the FC, which is a function of the sampling times $\mathbf{t}$. Since $\mathbf{x}(\mathbf{t}, {k})$ is a Gaussian random process with the same covariance matrices under both hypotheses, it can be seen that
\begin{equation}
    \mathbb{E}[\Lambda(\mathbf{t}, i)|\mathcal{H}_1] = -\mathbb{E}[\Lambda(\mathbf{t}, i)|\mathcal{H}_0] = \frac{1}{2} \mathbf{\tilde{s}}^{\text{T}} \mathbf{\Sigma}^{-1} (\mathbf{t}) \thinspace \mathbf{\tilde{s}},
    \label{eq:exp_LLR_KLD}
\end{equation}
where $\mathbf{t} = [t_1, t_2, \dots, t_N]$ are the sampling times at each of the sensors relative to the period at which the samples are obtained. It can be seen that the stopping times in (\ref{eq:exp_stop_time_H1}) and (\ref{eq:exp_stop_time_H0}) are inversely proportional to the respective expected values of LLRs in (\ref{eq:exp_LLR_KLD}), which is also the KL divergence (KLD) between the probability density functions of the observations under $\mathcal{H}_1$ and $\mathcal{H}_0$.

\subsection{Covariance Matrix and Correlation Functions }

The noise $\mathbf{n}(\mathbf{t}, k)$ being a Gaussian process can be characterized by a valid (positive definite) covariance matrix. Further, if we define pairwise correlation functions as $r_{j}(t_{fc}, t_j)$ which are a function of the sampling times, it must be ensured that $\mathbf{\Sigma}(\mathbf{t})$ remains positive definite ($\mathbf{\Sigma}(\mathbf{t}) \succ 0$ denotes that $\mathbf{\Sigma}(\mathbf{t})$ is positive definite) for all values of $t_1, t_2, \dots, t_N$. 



We assume that the pairwise correlations between observations $x_{fc}(t_{fc})$ and $x_{j}(t_j)$ are
\begin{equation}
    r_{j}(t_{fc}, t_j) = \rho_{j} f_{j}(t_{fc}, t_j).
\end{equation}

Since the sensor observations $\mathbf{x}_j(\cdot), \thinspace \forall j$ and $x_{fc}(\cdot)$ are jointly WSS processes, we have $f_{j}(t_{fc}, t_j) = f_{j}(|t_{fc} - t_j|)$, with $0<f_{j}(|t_{fc}-t_j| \leq 1$. 

\begin{asu}
The correlation functions $f_j(\cdot)$ are \textit{decreasing} functions of $|t_{fc}-t_j|$ with $f_j(0) = 1$.  
\label{as:dec_corr_fns}
\end{asu}
Some examples of such correlation functions are $e^{-|t_{fc}-t_j|}$ and $e^{-(t_{fc}-t_j)^2}$, which are all decreasing functions of $|t_{fc}-t_j|$ . 

The covariance matrix corresponding to the correlation structure considered in this work can be written as

\begin{equation}
\mathbf{\Sigma}(t) = \sigma^2
\begin{bmatrix}
\mathbf{I}_{N}  & \mathbf{r}_{fc}(\mathbf{t}) \\[0.5em]
\mathbf{r}^{T}_{fc}(\mathbf{t}) & 1
\end{bmatrix},
\label{eq:cov_structure_mat}
\end{equation}
where $\mathbf{I}_{N} \in \mathbb{R}^{N \times N}$ is the identity matrix and $\mathbf{r}_{fc}{(\mathbf{t})} = [r_1(\mathbf{t}) \dots r_{N}(\mathbf{t})]^T \in \mathbb{R}^{1 \times N}$, and the pair-wise correlations between each of the sensor observations and the observation at the FC can be represented as $r_j(\mathbf{t}) = \rho_j f_j(|t_{fc}-t_j|)$ for $j=1, \dots, N$ \footnote{This situation occurs when a low-resolution sensor observes a large region of interest (ROI) and multiple high-resolution sensors observe disjoint regions within the ROI.}. The correlation of the FC's observation with the $j^{th}$ sensor's observation when $t_{fc} = t_j$, i.e., when the $j^{th}$ sensor's observation is \textit{synchronous} with the FC's observation is $\rho_j$.  


We define $\mathbf{\Sigma}_{eq}$ as the covariance matrix that is formed when $t_{fc} = t_1 = \dots = t_{N}$, i.e., when all the sensors transmit \textit{synchronously}, and $\mathbf{r}_{eq}$ is the vector of pair-wise correlations of the fusion center with the observations at the other sensors. 

\begin{lem}
If $\mathbf{\Sigma}_{eq} \succ 0$, then $\mathbf{\Sigma}(t) \succ 0 \thinspace \forall \thinspace t_1, \dots, t_{N}, t_{fc} \in \mathbb{R}^{+}$. 
\label{lem:corr_sum_constraint}
\end{lem}
\textit{Proof: } If $\mathbf{\Sigma}_{eq} \succ 0$, we have $1-\mathbf{r}^{T}_{eq}\mathbf{r}_{eq} > 0$ through the Schur complement lemma. This is equivalent to $||\mathbf{r}_{eq}||_{2}^{2}<1$, which is further equivalent to $\sum_{j=1}^{N} \rho^2_{j} < 1$. For the condition $\mathbf{\Sigma}(\mathbf{t}) \succ 0$ to be satisfied, it is required that $\sum_{j=1}^{N} \rho^2_{j} f^{2}_{j} < 1$, according to the Schur complement lemma. This is satisfied due to $0<f_{j} \leq 1$ from Assumption \ref{as:dec_corr_fns}, which leads to $\sum_{j=1}^{N} \rho^2_{j} < 1$. This completes the proof. $\blacksquare$
\vskip 0.1cm
Since it has been established that $\mathbf{\Sigma}(\mathbf{t}) \succ 0$, we can expand the KL divergence term in (\ref{eq:exp_LLR_KLD}) using block matrix inversion. For a block matrix of the form $\mathbf{\Sigma}(\mathbf{t})$,  
\begin{equation}
\mathbf{\Sigma}^{-1}(\mathbf{t}) = \frac{1}{\sigma^2} 
\begin{bmatrix}
\mathbf{I}_{N} + \frac{1}{d_{fc}(\mathbf{t})} \mathbf{r}_{fc}(\mathbf{t}) \mathbf{r}^{T}_{fc}(\mathbf{t}) & \frac{-1}{d_{fc}(\mathbf{t})}\mathbf{r}_{fc}(\mathbf{t}) \\[1em]
\frac{-1}{d_{fc}(\mathbf{t})}\mathbf{r}^{T}_{fc}(\mathbf{t}) & \frac{1}{d_{fc}(\mathbf{t})}
\end{bmatrix}
\label{eq:inv_cov}
\end{equation}
where $d_{fc}(\mathbf{t}) = 1 - \mathbf{r}^{T}_{fc}(\mathbf{t})\mathbf{r}_{fc}(\mathbf{t})$. 


\section{Analysis of expected stopping time of SPRT}
The KL divergence term in (\ref{eq:exp_LLR_KLD}) $\frac{1}{2} \mathbf{\tilde{s}}^{\text{T}} \mathbf{\Sigma}^{-1}(\mathbf{t}) \mathbf{\tilde{s}}$ can be simplified as 
\begin{align}
  \frac{1}{2} \mathbf{\tilde{s}}^{\text{T}} \mathbf{\Sigma}^{-1}(\mathbf{t}) \mathbf{\tilde{s}} = \frac{||\mathbf{s}_{N}||^2}{2\sigma^2 } + \frac{1}{2\sigma^2} \frac{(\sum_{j=1}^{N}s_j r_j(\mathbf{t})-s_{fc})^2}{1-\sum_{i=1}^{N} r^{2}_{j}(\mathbf{t})}.
  \label{eq:KL_div_expansion}
\end{align}
The first term on the right-hand side (RHS) of \eqref{eq:KL_div_expansion} is independent of the sampling times in $\mathbf{t}$. The second term contains the sampling times and it is this term that affects the expected stopping time of the sequential test. For notational convenience, we denote $\mathbf{\tilde{s}}^{\text{T}} \mathbf{\Sigma}^{-1}(\mathbf{t}) \mathbf{\tilde{s}}$ by $\mathbb{K}(\mathbf{t})$, $(\sum_{j=1}^{N}s_j r_j(\mathbf{t})-s_{fc})^2$ as $\mathbb{N}(\mathbf{t})$ and $1-\sum_{j=1}^{N} r^{2}_{j}(\mathbf{t})$ as $\mathbb{D}(\mathbf{t})$.

\subsection{Upper and lower bounds on KL divergence}
An upper bound on $\mathbb{K}(\mathbf{t})$ can be obtained according to,
\begin{align}
   \mathbb{K}(\mathbf{t}) < \frac{||\mathbf{s}||^2}{2\sigma^2 } + \frac{1}{2\sigma^2} \frac{\underset{\mathbf{t}}{\text{max}} \thinspace \thinspace {\mathbb{N}(\mathbf{t})}}{\underset{\mathbf{t}}{\text{min}} \thinspace \thinspace {\mathbb{D}(\mathbf{t})}}.
   \label{eq:KL_UB}
\end{align}
From the expression above, we first consider $\underset{t}{\text{min}} \thinspace \thinspace {\mathbb{D}(\mathbf{t})}$. This can be found as
\begin{align}
 \underset{\mathbf{t}}{\text{min}} \thinspace \thinspace {\mathbb{D}(\mathbf{t})} =  1 - ||\mathbf{r}_{eq}||^{2}_2,
 \label{eq:D_LB}
\end{align}
which is due to the fact that $f_j(|t_{fc}-t_j|)$ is a decreasing function of $|t_j - t_{fc}|$ and $f_j(0) = 1$, due to Assumption \ref{as:dec_corr_fns}.  

Next, we aim to obtain $\underset{\mathbf{t}}{\text{max}} \thinspace \thinspace {\mathbb{N}(\mathbf{t})}$ to complete the analysis of the upper bound on $\mathbb{K}(\mathbf{t})$. We first define two sets $\mathcal{A}^{+}$ and $\mathcal{A}^{-}$ as
\begin{equation}
    \mathcal{A}^{+} = \bigg \{j \thinspace : s_j \rho_j >0, \thinspace \forall \thinspace j = 1, \dots, N \bigg \}, \thinspace \text{and}
    \label{eq:pos_terms_KL_UB}
\end{equation}
\begin{equation}
      \mathcal{A}^{-} = \bigg \{j \thinspace : s_j \rho_j <0, \thinspace \forall \thinspace j = 1, \dots, N       \bigg \}.  
      \label{eq:neg_terms_KL_UB}
\end{equation}
Upon maximizing $\mathbb{N}(\mathbf{t})$, we have the following proposition.
\begin{proposition}
The upper bound of $\mathbb{K}(t)$, denoted by $\bar{\mathbb{B}}$, is  
\label{prop:kl_up_b}
\end{proposition}
\begin{align}
   \mathbb{\overline{B}} = & \frac{1}{2 \sigma^2\cdot d_{\text{eq}}}\textrm{max} \Bigg \{ { \bigg (\Psi^{+}(\epsilon) -s_{fc} \bigg )^2 } \nonumber, {\bigg( \Psi^{-}(\epsilon) -s_{fc} \bigg )^2 }\Bigg \} \nonumber\\
      &+ \frac{||\mathbf{s}||^2}{2\sigma^2}, 
     \label{eq:KL_UB_full_exp}
\end{align}
where the quantities $d_{\text{eq}} = (1 - ||\mathbf{r}_{\text{eq}}||^{2}_2)$, $\delta_{fc}(\epsilon) = \text{max} \thinspace \{t_{fc}, |t_{fc}-\epsilon| \}$, $\Psi^{+}(\epsilon)= \sum \limits_{j \in \mathcal{A}^{+}}s_j\rho_j +  \sum \limits_{j \in \mathcal{A}^{-}}s_j\rho_j f_{j}(\delta_{fc}(\epsilon) )$ and, $\Psi^{-}(\epsilon) = \sum \limits_{j \in \mathcal{A}^{+}} s_j\rho_j f_{j}(\delta_{fc}(\epsilon)) +  \sum \limits_{j \in \mathcal{A}^{-}}s_j\rho_j$.
\begin{proof}
Please refer to Appendix A.
\end{proof}
To establish the lower bound, we define the sets
\begin{equation}
    \mathcal{B}^{+} = \bigg \{j \thinspace : s_{fc} s_j \rho_j >0, \thinspace \forall \thinspace j = 1, \dots, N \bigg \}, \thinspace \text{and}
    \label{eq:pos_terms_KL_LB}
\end{equation}
\begin{equation}
      \hspace{-1.1cm}\mathcal{B}^{-} = \bigg \{j \thinspace : s_{fc} s_j \rho_j <0, \thinspace \forall \thinspace j = 1, \dots, N       \bigg \}.  
      \label{eq:neg_terms_KL_LB}
\end{equation}
Then, the following proposition establishes a lower bound on $\mathbb{K}(\mathbf{t})$:
\begin{proposition}
The lower bound of $\mathbb{K}(t)$ is given as
\begin{equation}
    \underline{\mathbb{B}} = \frac{||\mathbf{\tilde{s}}||^{4}/{(2 \sigma^2)}}{||\mathbf{\tilde{s}}||^{2} +  2 \cdot s_{fc} \cdot \bigg[ \sum_{j \in \mathcal{B}^{-}} s_j \rho_j f_j(\delta_{fc}(\epsilon)) + \sum_{j \in \mathcal{B}^{+}} s_j \rho_j) \bigg]}
    \label{eq:KL_LB_full_exp}
\end{equation}
\end{proposition}

\begin{proof}
Please refer to Appendix B.
\end{proof}
We would like to remark that in both the upper and lower bounds obtained in (\ref{eq:KL_UB_full_exp}) and (\ref{eq:KL_LB_full_exp}), we assume that $t_{fc}$ is fixed. 

In order to numerically evaluate the maximum and minimum values of the KL-divergence term in (\ref{eq:KL_div_expansion}), we use the non-linear optimization solver \texttt{fmincon} in \MATLAB, which uses interior-point methods to obtain the solutions corresponding to the maximum and minimum. We use the \texttt{MultiStart} method, which initializes the solver at different starting points, so as to increase the likelihood of obtaining a globally optimal solution \footnote{We would like to remark that the solver does not always guarantee a globally optimal solution of (\ref{eq:KL_div_expansion}).}. 


In Fig. \ref{fig:bounds_min_max_num_sens}, the pairwise correlations between the sensor observations and the FC's observations are such that $|\rho_j| =\frac{1}{2N+1}$, where $N$ is the number of sensors, and the correlation functions $f_j(|t_{fc}-t_j|) = e^{-(t_{fc}-t_j)^2}, \forall j=1, \ldots, N$. The number of sensors with positive correlations and negative correlations are both equal to $\frac{N}{2}$. The sampling window is considered to be of duration $\epsilon = 1$ units of time and $t_{fc}$ is set at 0. \color{black} All the plots for the expected stopping time in Fig. \ref{fig:bounds_min_max_num_sens} are obtained when the true hypothesis is $\mathcal{H}_0$.

\begin{figure}[!h]
    \centering
    \includegraphics[clip, trim=0.5cm 5.5cm 0.5cm 5.5cm,width = 0.65 \linewidth]{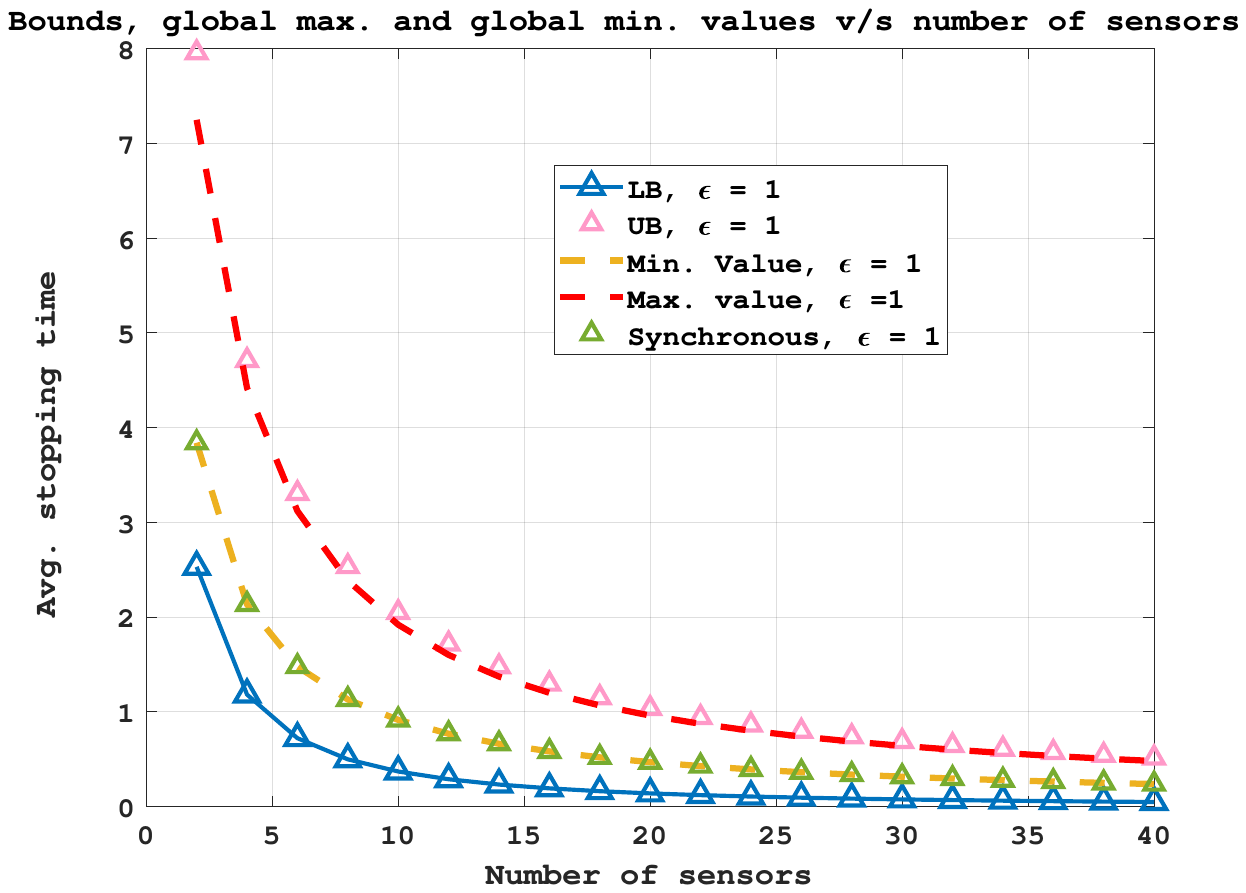}
    \caption{Expected stopping time versus the number of sensors for $s_{fc} = s_j = 0.5$ and $\sigma^2 = 1$; The figure is used to illustrate the nature of the upper and lower bounds and the maximum and minimum values of the stopping time as a function of the number of sensors when the prob. of detection $P_D = 0.92$ and the prob. of false-alarm $P_{FA} = 0.1$ for the SPRT.}
    \label{fig:bounds_min_max_num_sens}
\end{figure}

It can be seen from Fig. \ref{fig:bounds_min_max_num_sens} that the upper bound and maximum values are in proximity to each other. With an increasing number of sensors, the minimum value and the lower bound grow increasingly closer to one another. This demonstrates that the bounds are reasonably close to the numerical evaluations of the expected value for systems with a different number of sensors. The decrease in the expected stopping time as the number of sensors increases can be explained by noting that as the number of sensors increases, the term $||\mathbf{s}_{N}||^2$ in (\ref{eq:KL_div_expansion}) increases. Also, the absolute values of the correlations $|\rho_j| =\frac{1}{2N+1}$ decrease with $N$ so that the constraint on the sum of squares of correlation coefficients $\sum_{j=1}^{N} \rho_{j}^{2} < 1$ from Lemma \ref{lem:corr_sum_constraint} is satisfied. 

\begin{figure}
    \centering
\begin{subfigure}
    \centering
    \includegraphics[clip, trim=0.5cm 0.5cm 0.5cm 0.5cm,width =0.7\linewidth]{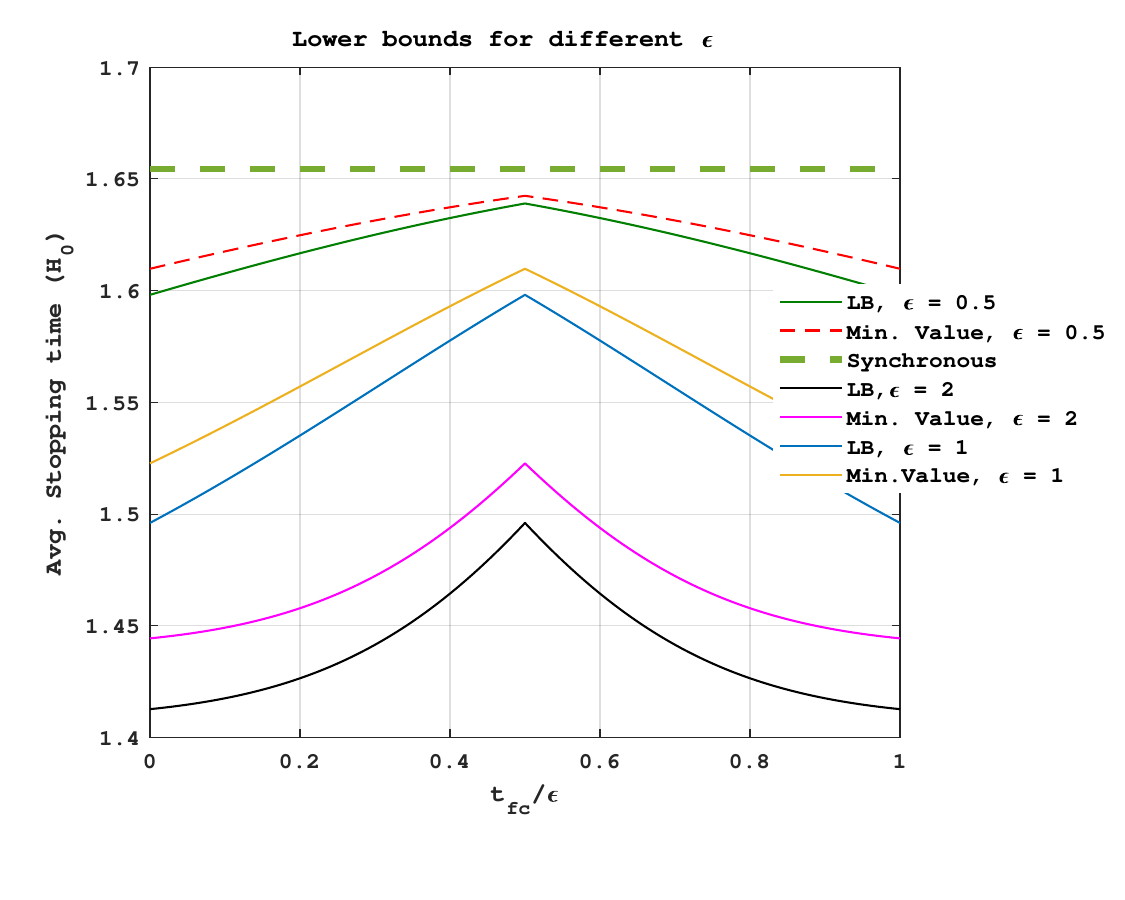}
    \label{subfig:overlap_LB_min}
\end{subfigure}
\hspace{\fill}
\begin{subfigure}
    \centering
    \includegraphics[clip, trim=0.5cm 6cm 0.5cm 6cm,width = 0.8\linewidth]{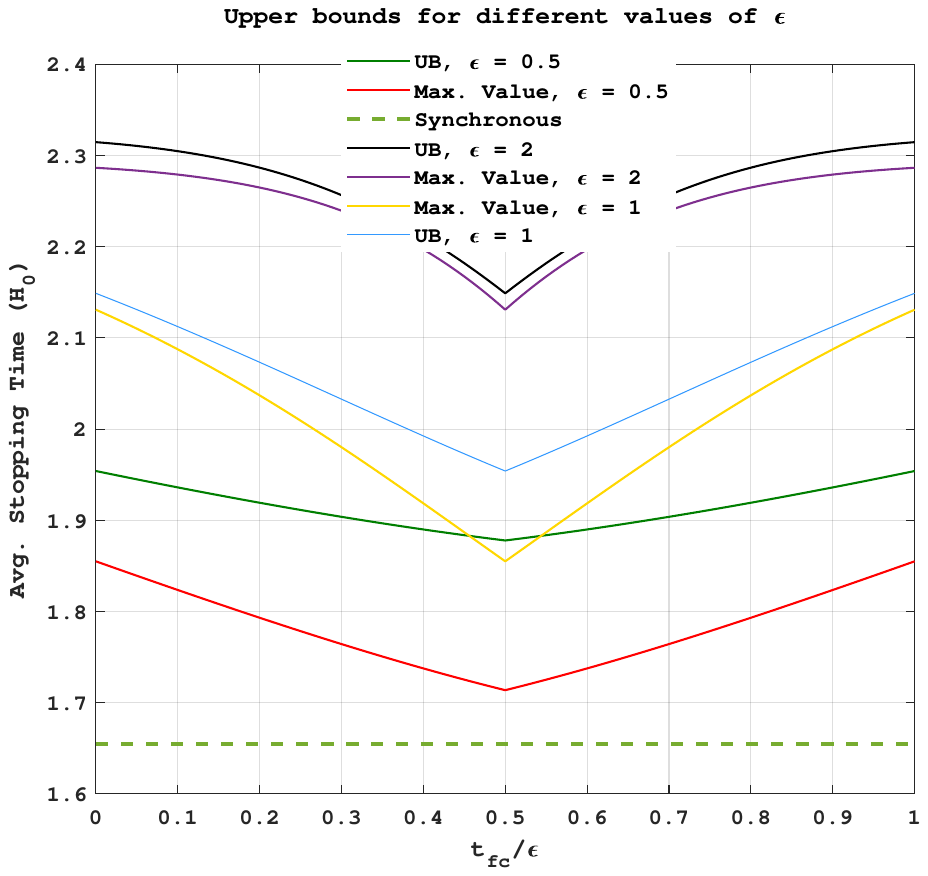}
    \label{subfig:overlap_UB_max}
\end{subfigure} 
    \caption{Expected stopping time versus the sampling time at the FC for $s_{fc} = s_i = 0.5$ and $\sigma^2 = 1$, when the prob. of detection $P_D = 0.92$ and the prob. of false-alarm $P_{FA} = 0.1$ for the SPRT.}
    \label{fig:bounds_two_figs}
\end{figure} 

In Fig. \ref{fig:bounds_two_figs}, we illustrate the variation of the bounds, global maximum, and global minimum values as a function of sampling times $t_{fc}$ for a system with 8 sensors. The observations of 4 sensors among the 8 sensors are positively correlated with FC's observations, with $\rho_j = 0.1$, and the other 4 sensors are negatively correlated with $\rho_j = -0.25$. The correlation functions are $f_j(|t_{fc}-t_j|) = e^{-(t_{fc}-t_j)^2}, \forall j=1, \ldots, 8$. The plots are normalized by considering $\frac{t_{fc}}{\epsilon}$ for the respective sampling windows $\epsilon = \{0.5, 1, 2\}$ to facilitate easier comparisons.  

From the curves in Fig. \ref{fig:bounds_two_figs}, we notice that the upper bound and maximum as a function of $t_{fc}$ attain the least value at $\frac{t_{fc}}{\epsilon} = 0.5$, for all values of $\epsilon$ considered in these figures. For the lower bound and minimum values, the maximum value is attained at $t_{fc} = \frac{\epsilon}{2}$. The curves for the upper bound and maximum increase in value as $\epsilon$ is larger (and decrease for the lower bound and minimum values). These observations can be explained by noting that when $\frac{t_{fc}}{\epsilon} = 0.5$, $\delta_{fc}(\epsilon) = 0.5$ is the smallest value and the optimal solutions $\hat{t}_{j}$ and $\check{t}_{j}$ are restricted to a smaller neighborhood around $t_{fc}$ ($|t_{fc}-\hat{t}_{j}|$, $|t_{fc}-\check{t}_{j}| \leq \delta_{fc}(\epsilon) = 0.5$). When the constraint set $0 \leq t_j \leq \epsilon$ is larger due to a larger value of $\epsilon$,  the optimal values may deviate further away from $t_{fc}$. Also, for a fixed value of $\epsilon$, if $t_{fc} = 0$ or $t_{fc}= \epsilon$, the optimal values are able to deviate further away from $t_{fc}$.

\section{Conclusion}
In this work, we presented a study of sampling in a distributed sequential hypothesis testing problem, when there are correlations between the observations at the sensors and the FC. We analyzed the performance of the system through the expected stopping time of the SPRT when the sensors were sampled at different time instants. We presented bounds on the expected stopping time of the sequential test and demonstrated the nature of these bounds and the nature of the maximum and minimum values of the expected stopping time as a function of the sampling time at the FC, using numerical results.

\bibliographystyle{IEEEtran}
\bibliography{sampling_SPRT.bib}

\section{Appendix A:\textbf{Proof of upper bound on KL divergence}}
We first define two sets $\mathcal{A}^{+}$ and $\mathcal{A}^{-}$ as
\begin{equation}
    \mathcal{A}^{+} = \bigg \{j \thinspace : s_j \rho_j >0, \thinspace \forall \thinspace j = 1, \dots, N \bigg \}, \thinspace \text{and}
\end{equation}
\begin{equation}
      \mathcal{A}^{-} = \bigg \{j \thinspace : s_j \rho_j <0, \thinspace \forall \thinspace j = 1, \dots, N       \bigg \}.  
\end{equation}
The maximization of ${\mathbb{N}(\mathbf{\mathbf{t}})}$ can be reformulated as: 
\begin{align}
   \underset{\mathbf{t}}{\text{max}} \thinspace \thinspace {\mathbb{N}(\mathbf{t})} &= \thinspace  {\text{max}} \thinspace \thinspace  \bigg \{ \bigg( \underset{\mathbf{t}}{\text{max}} \thinspace \thinspace \mathbb{L}(\mathbf{t}) \bigg)^2, \thinspace \thinspace \bigg(\underset{\mathbf{t}}{\text{min}} \thinspace \thinspace \mathbb{L}(\mathbf{t}) \bigg)^2 \thinspace \bigg \}.
   \label{eq:N_t_max}
\end{align}
where $\mathbb{L}(\mathbf{t}) = \thinspace \sum_{j \in \mathcal{A}^{+}} s_j r_j(\mathbf{t}) + \sum_{j \in \mathcal{A}^{-}} s_j r_j(\mathbf{t})-s_{fc}$. 
\vskip 0.2cm
We also define the following quantities, $\mathbb{L}^{\Uparrow}= \underset{\mathbf{t}}{\text{max}} \thinspace \thinspace \Big[\mathbb{L}(\mathbf{t}) \Big]$ and, 
$\mathbb{L}^{\Downarrow}= \underset{\mathbf{t}}{\text{min}} \thinspace \thinspace \Big[\mathbb{L}(\mathbf{t}) \Big]$. The following two scenarios are examined in order to determine $\mathbb{L}^{\Uparrow}(\mathbf{t})$ and $\mathbb{L}^{\Downarrow}(\mathbf{t})$:

 \vskip 0.2cm
\begin{itemize}
\item \underline{\textbf{Case 1}}: If $|\mathbb{L}^{\Uparrow}| > |\mathbb{L}^{\Downarrow}|$, then $\mathbb{L}^{\Uparrow} > 0$ and therefore, $\underset{\mathbf{t}}{\text{max}} \thinspace \thinspace {\mathbb{N}(\mathbf{t})} =  \bigg( \underset{\mathbf{t}}{\text{max}} \thinspace \thinspace \mathbb{L}(\mathbf{t}) \bigg)^2 = (\mathbb{L}^{\Uparrow})^2$. Then,  
\begin{align}
  \mathbb{L}^{\Uparrow}  &= \thinspace \underset{\mathbf{t}}{\text{max}} \bigg [\thinspace \sum_{j \in \mathcal{A}^{+}} s_j r_j(\mathbf{t}) + \sum_{j \in \mathcal{A}^{-}} s_j r_j(\mathbf{t})-s_{fc} \bigg ]  \nonumber \\
  &= \thinspace \thinspace \sum_{j \in \mathcal{A}^{+}} s_j \rho_j \bigg[ \underset{t_j}{\text{max}} \thinspace f_j(|t_{fc}-t_j|)  \bigg] -s_{fc} \nonumber\\
  & \quad + \sum_{j \in \mathcal{A}^{-}} s_j \rho_j \bigg[ \underset{t_j}{\text{min}} \thinspace f_j(|t_{fc}-t_j|)  \bigg] \nonumber\\
  &= \thinspace \thinspace \sum_{j \in \mathcal{A}^{+}} s_j \rho_j  +  \sum_{j \in \mathcal{A}^{-}} s_j \rho_j f_{j}(\text{max} \thinspace \{t_{fc}, |t_{fc}-\epsilon| \}) -s_{fc}. 
  \label{eq:N_u_t_sol}
\end{align}
\vskip 0.1cm
\item \underline{\textbf{Case 2}}: If $|\mathbb{L}^{\Uparrow}| < |\mathbb{L}^{\Downarrow}|$, then $\mathbb{L}^{\Downarrow} < 0$ and therefore, $\underset{\mathbf{t}}{\text{max}} \thinspace \thinspace {\mathbb{N}(\mathbf{t})} =  \bigg( \underset{\mathbf{t}}{\text{min}} \thinspace \thinspace \mathbb{L}(\mathbf{t}) \bigg)^2 = (\mathbb{L}^{\Downarrow})^2$. Then,
\begin{align}
  \mathbb{L}^{\Downarrow} &= \thinspace \underset{\mathbf{t}}{\text{min}} \bigg [\thinspace \sum_{j \in \mathcal{A}^{+}} s_j r_j(\mathbf{t}) + \sum_{j \in \mathcal{A}^{-}} s_j r_j(\mathbf{t})-s_{fc} \bigg ] \nonumber \\
  &= \thinspace \thinspace \sum_{j \in \mathcal{A}^{+}}s_j \rho_j \bigg[ \underset{t_j}{\text{min}} \thinspace f_j(|t_{fc}-t_j|)  \bigg]  \nonumber\\
  & \quad + \sum_{j \in \mathcal{A}^{-}}s_j \rho_j \bigg[ \underset{t_{fc}, t_j}{\text{max}} \thinspace f_j(|t_{fc}-t_i|)  \bigg]-s_{fc} \nonumber\\
  &= \thinspace \thinspace \sum_{j \in \mathcal{A}^{+}}s_j \rho_j f_{j}(\text{max} \thinspace \{t_{fc}, |t_{fc}-\epsilon| \}) \nonumber\\
  & \quad +  \sum_{j \in \mathcal{A}^{-}}s_j \rho_j -s_{fc}.
  \label{eq:N_l_t_sol}
\end{align}
\end{itemize}

When $\underset{\mathbf{t}}{\text{max}} \thinspace \thinspace {\mathbb{N}(\mathbf{t})} = (\mathbb{L}^{\Uparrow})^2$, the expression in (\ref{eq:N_u_t_sol}) can be explained by noting that in (\ref{eq:N_u_t_sol}), the summation consisting of terms from $\mathcal{A}^{+}$ (which are all positive) attains its maximum when $f_j(\cdot) = 1$. Similarly, the second summation contains only negative terms, all of which belong to $\mathcal{A}^{-}$ and attain its maximum value when $f_j(\cdot) = f_j(\underset{t_j}{\text{max}} \thinspace |t_{fc}-t_j|)$. When $t_{fc}$ is fixed, $\text{max} \thinspace |t_{fc}-t_j|$ occurs when either $t_j = 0$ or $\epsilon$, and since $f_j(\cdot)$ is a decreasing function from Assumption \ref{as:dec_corr_fns}, and minimizes the magnitude of the negative terms in $\mathcal{A}^{-}$. The expression in (\ref{eq:N_l_t_sol}) can also be obtained through a similar form of analysis. By using the expressions in (\ref{eq:N_l_t_sol}) and (\ref{eq:N_u_t_sol}), which are solutions to (\ref{eq:N_t_max}) in (\ref{eq:KL_UB}), an upper bound on $\mathbb{K}(\mathbf{t})$ is established. We denote the upper bound of $\mathbb{K}(\mathbf{t})$ by $\bar{\mathbb{B}}$, the quantity $\text{max} \thinspace \{t_{fc}, |t_{fc}-\epsilon| \}$ by $\delta_{fc}(\epsilon) $, and upon combining the expressions in (\ref{eq:N_t_max}), (\ref{eq:N_u_t_sol}) and (\ref{eq:N_l_t_sol}) the upper bound is given as 

\begin{align}
   \mathbb{\overline{B}} = & {\text{max}} \Bigg \{ \frac{ \bigg (\sum \limits_{j \in \mathcal{A}^{+}} s_j \rho_j +  \sum \limits_{j \in \mathcal{A}^{-}} s_j \rho_j f_{j}(\delta_{fc}(\epsilon) ) -s_{fc} \bigg )^2 }{2 \sigma^2\cdot(1 - ||\mathbf{r}_{fc}||^{2}_2)} \nonumber \\
      &, \frac{\bigg( \sum \limits_{j \in \mathcal{A}^{+}}s_j \rho_j f_{j}(\delta_{fc}(\epsilon)) +  \sum \limits_{j \in \mathcal{A}^{-}}s_j \rho_j -s_{fc} \bigg )^2 }{2 \sigma^2 \cdot (1 - ||\mathbf{r}_{fc}||^{2}_2)}\Bigg \} \nonumber\\
      &+ \frac{||\mathbf{s}||^2}{2\sigma^2}. \blacksquare
      \label{eq:UB_proof_eq_app}
\end{align}

We would like to remark that the result in (\ref{eq:UB_proof_eq_app}) can be obtained by also considering $\mathbb{P}(\mathbf{t}) = -\mathbb{L}(\mathbf{t})$ and proceeding in a similar manner as the analyses in the two cases that resulted in (\ref{eq:N_u_t_sol}) and (\ref{eq:N_l_t_sol}).

\section{Appendix B:\textbf{Proof of lower bound on KL divergence}}

A lower bound of $\mathbb{K}(t)$ can be derived by using the \textit{Cauchy-Schwarz} inequality on the dot product of $\mathbf{\Sigma}^{-\frac{1}{2}}(\mathbf{t}) \thinspace \mathbf{\tilde{s}}$ and $\mathbf{\Sigma}^{\frac{1}{2}}(\mathbf{t}) \thinspace \mathbf{\tilde{s}}$ as follows, 

\begin{align}
    &(\mathbf{\Sigma}^{-\frac{1}{2}}(\mathbf{t}) \mathbf{\tilde{s}})^{T} \cdot (\mathbf{\Sigma}^{\frac{1}{2}}(\mathbf{t}) \mathbf{\tilde{s}}) \thinspace \leq \thinspace ||\mathbf{\Sigma}^{-\frac{1}{2}}\mathbf{\tilde{s}}|| \thinspace \cdot \thinspace ||\mathbf{\Sigma}^{\frac{1}{2}}\mathbf{\tilde{s}}|| \nonumber \\
    \implies & ||\mathbf{\tilde{s}}||^{2} \thinspace \thinspace \leq \thinspace \thinspace ||\mathbf{\Sigma}^{-\frac{1}{2}}\mathbf{\tilde{s}}|| \thinspace \cdot \thinspace ||\mathbf{\Sigma}^{\frac{1}{2}}\mathbf{\tilde{s}}|| \nonumber\\
    \implies & ||\mathbf{\tilde{s}}||^{2} \thinspace \thinspace \leq \thinspace \thinspace (\mathbf{\tilde{s}}^{\text{T}} \mathbf{\Sigma}^{-1}(\mathbf{t}) \mathbf{\tilde{s}})^{\frac{1}{2}}\cdot(\mathbf{\tilde{s}}^{\text{T}} \mathbf{\Sigma}(\mathbf{t})\mathbf{\tilde{s}})^{\frac{1}{2}} \nonumber \\
    \implies & \mathbb{K}(\mathbf{t}) = \frac{1}{2} \mathbf{\tilde{s}}^{\text{T}} \mathbf{\Sigma}^{-1}(\mathbf{t}) \mathbf{\tilde{s}}  \thinspace \thinspace \geq \thinspace \thinspace \frac{1}{2} \thinspace\frac{||\mathbf{\tilde{s}}||^{4}}{\mathbf{\tilde{s}}^{\text{T}} \mathbf{\Sigma}(\mathbf{t})\mathbf{\tilde{s}}} \nonumber \\
    \implies & \mathbb{K}(\mathbf{t}) \thinspace \geq \thinspace \thinspace \frac{1}{2} \thinspace \frac{||\mathbf{\tilde{s}}||^{4}}{\underset{\mathbf{t}}{\text{max}} \thinspace \mathbf{\tilde{s}}^{\text{T}} \mathbf{\Sigma}(\mathbf{t})\mathbf{\tilde{s}}}
    \label{eq:Cauchy_LB}
\end{align}
where $\mathbf{\Sigma}^{1/2}(\mathbf{t})$ is the square root of $\mathbf{\Sigma}(\mathbf{t})$, whose existence is guaranteed due to the positive definiteness of $\mathbf{\Sigma}(\mathbf{t})$ from Lemma 1.  

We further establish the behavior of the term $\underset{t}{\text{max}} \thinspace \mathbf{\tilde{s}}^{\text{T}} \mathbf{\Sigma}(\mathbf{t})\mathbf{\tilde{s}}$ by expanding it as,

\begin{align}
    \underset{t}{\text{max}} \thinspace \mathbf{\tilde{s}}^{\text{T}} \mathbf{\Sigma}(\mathbf{t})\mathbf{\tilde{s}} = \thinspace \underset{t}{\text{max}} \thinspace \thinspace s^2_{fc} &+ \sum_{j=1}^{N} s^2_j  \nonumber\\
    &+ 2 \cdot \sum_{j=1}^{N} s_j s_{fc} r_j f_j(|t_{fc}-t_j|). 
    \label{eq:no_inv_qf}
\end{align}
We divide the terms $s_{fc}s_j \rho_j f_j(|t_{fc}-t_j|)$ into two sets $\mathcal{B}^{+}$ and $\mathcal{B}^{-}$ where

\begin{equation}
    \mathcal{B}^{+} = \bigg \{m_j \thinspace : m_j = s_{fc} s_j \rho_j >0, \thinspace \forall \thinspace j = 1, \dots, N       \bigg \}, \thinspace \text{and}
\end{equation}

\begin{equation}
      \mathcal{B}^{-} = \bigg \{m_j \thinspace : m_j = s_{fc} s_j \rho_j <0, \thinspace \forall \thinspace j = 1, \dots, N       \bigg \}.  
\end{equation}
We consider the summation term $2 \cdot \sum_{j=1}^{N} s_j s_{fc} \rho_j f_j(|t_{fc}-t_j|)$ in (\ref{eq:no_inv_qf}), which can be written as 
\begin{equation}
 2 \cdot \sum_{j \in \mathcal{B}^{+}} s_j s_{fc} \rho_j f_j(|t_{fc}-t_j|) + 2 \cdot \sum_{j \in \mathcal{B}^{-}} s_j s_{fc} \rho_j f_j(|t_{fc}-t_j|). \label{eq:split_eq_LB} 
\end{equation}

The expression in (\ref{eq:no_inv_qf}) is always positive due to the fact that $\mathbf{\Sigma}(\mathbf{t})$is positive definite for all $\mathbf{t}$ (Lemma 1). The two terms in (\ref{eq:split_eq_LB}) are of opposite signs due to which the elements corresponding to $\mathcal{B}^{+}$ can be maximized by setting $t_{fc} = t_j$ and similarly the elements corresponding to $\mathcal{B}^{-}$ can be maximized by setting $t_j$ such that $|t_{fc}-t_{j}|$ is maximized. 

An important remark about the lower bound is that is not trivial (negative), due to Lemma 1. This is because the quantity $\mathbf{\tilde{s}}^{\text{T}} \mathbf{\Sigma}(t) \mathbf{\tilde{s}}$ is positive for all values of $0 \leq t_1, t_2, \dots, t_N, t_{fc} \leq \epsilon$.  

The complete expression for the lower bound of $\mathbb{K}(t)$, using the expressions in (\ref{eq:Cauchy_LB}), (\ref{eq:no_inv_qf}) and (\ref{eq:split_eq_LB}) is 

\begin{equation}
    \underline{\mathbb{B}} = \frac{||\mathbf{\tilde{s}}||^{4}/{(2 \sigma^2)}}{||\mathbf{\tilde{s}}||^{2} +  2 \cdot s_{fc} \cdot \bigg[ \sum_{j \in \mathcal{B}^{-}} s_j \rho_j f_j(\delta_{fc}(\epsilon)) + \sum_{j \in \mathcal{B}^{+}} s_j \rho_j) \bigg]}, 
\end{equation}
where $\delta_{fc}(\epsilon) = \text{max}\{t_{fc}, |t_{fc}-\epsilon|\}$.

\end{document}